\documentclass{aamas2016}
\usepackage[usenames, dvipsnames]{color}
\usepackage{amsmath}

\definecolor{purple}{RGB}{168,56,168}				

 \DeclareMathOperator*{\argmin}{arg\,min}
  \DeclareMathOperator*{\argmax}{arg\,max}
 \newtheorem{theorem}{Theorem}
 \newtheorem{lemma}{Lemma}
 \newtheorem{definition}{Definition}
 \newtheorem{example}{Example}
 \newtheorem{corollary}{Corollary}

\begin{document}


\title{Budgetary Effects on Pricing Equilibrium in Online Markets}



%
%
%
%

%

\numberofauthors{3}

\author{
%
\alignauthor
Allan Borodin\\
\affaddr{University of Toronto}\\
\affaddr{Toronto, Canada}\\
\email{bor@cs.toronto.edu}
\alignauthor
Omer Lev\\
\affaddr{University of Toronto}\\
\affaddr{Toronto, Canada}\\
\email{omerl@cs.toronto.edu}
\alignauthor
Tyrone Strangway\\
\affaddr{University of Toronto}\\
\affaddr{Toronto, Canada}\\
\email{tyrone@cs.toronto.edu}
}

\maketitle

\begin{abstract}
Following the work of Babaioff et al~\cite{BNL14}, we consider the pricing game 
with strategic vendors and a single buyer, modeling a scenario in which multiple competing vendors have very good knowledge of a buyer, as is common in online markets. We add to this model the 
realistic assumption that the buyer has a fixed budget and does not have 
unlimited funds. When the buyer's valuation function is additive, we are able 
to completely characterize the different possible pure Nash Equilibria (PNE) 
and in particular obtain a necessary and sufficient condition for 
uniqueness. Furthermore, we characterize the market 
clearing (or Walresian) equilibria for all submodular valuations.

Surprisingly, for certain monotone submodular function valuations, we show 
that the pure NE can exhibit some counterintuitive phenomena; namely, there is 
a valuation such that the pricing will be market clearing and within budget 
if the buyer does not reveal the budget but will result in a smaller set of 
allocated items (and higher prices for items) if the buyer does reveal the 
budget. It is also the case that the conditions that guarantee market clearing 
in Babaioff et al~\cite{BNL14} for submodular functions are not necessarily 
market clearing when there is a budget. Furthermore, with respect to social welfare, while without budgets all equilibria are optimal (i.e. POA = POS = 1), we show that with budgets the worst equilibrium may only achieve $\frac{1}{n-2}$ of the best equilibrium.

\end{abstract}


\category{K.4}{Computers and Society}{Electronic Commerce}
\category{J.4}{Social and Behavioral Sciences}{Economics}



\terms{Economics, Theory}


\keywords{pricing, price of anarchy, budget, Nash equilibrium}

\section{Introduction}


The question of how to price items in a market has existed for millennia, ever since people began trading with one another. Market pricing is clearly a central area of study in Economics. One of the most important aspects of these studies 
is market clearing (Walrasian) equilibria, initiated by Walras in 1874-1877 (See translation by Willam Jaffe \cite{WalrasJ03}). The advent of online markets has changed pricing in many respects, both to the vendors benefit and to their detriment. On the one hand, vendors can gain vast information about their customers, and can tailor their prices to each particular customer. On the other hand, the instant availability of a large amount of competitors allows customers to pick and choose among the items they want. The speed and size of online markets necessitates having simple rules for setting prices and understanding the nature of
pricing equilibria.  

This world, where vendors compete with each other in online marketplaces (e.g., Amazon or E-Bay), and are able to adapt their prices to their customer, inspired a recent line of research \cite{BNL14,LOBR15}, utilizing a game theoretical toolbox to address these scenarios. Here self interested agents (the vendors) post prices for their wares and a buyer consumes a subset of the items based on their valuation for the items and the posted prices. The vendors aim is to maximize their individual profits (we assume no transaction or production costs) and they will post prices as high as possible. If a vendor does not wish to participate in the market they may price their items at infinity. How vendors best position themselves in the market relative to other vendors -- and the valuation of the buyer -- creates interesting and non trivial scenarios. While we use buyer - seller terminology, this problem is applicable to other allocation scenarios that may arise in a multi-agent environment. For example one may envision a setting with potential employees (the vendors) and a company (the buyer). The company, subject to wage constraints, wishes to hire some subset of the employees. 

These works, and ours, are interested in the equilibria solutions of the games. 
In particular the focus is on the study of the pure Nash equilibrium (PNE). Informally a set of prices form a PNE if no single agent (vendor) can improve her position by modifying her price. We are interested in the properties that PNE poses. That is what items are purchased, what is the overall social welfare and when does the market clear.

We add to the setting of \cite{BNL14,LOBR15} the vital component of budgets. Previous work assumed the buyer had unlimited purchasing power, their decision limited only by the valuations. Now we cannot assume the buyer will purchase any bundle at a price lower than its value.
The buyer will consume the set of items with highest utility (i.e., net value, taking prices into account), that is within their budget, while the strategies of the vendors is determining the prices of the items they offer.

This brings about a departure from the previous results, as we show how the buyer's valuation and budget influence the market. Not only are some previous results now rendered impossible as they exceed the budget, but this more realistic assumption materially changes the structure of pricing from the intuitions formed in earlier research, as a more complex strategic behaviour needs to considered.


 
\subsection{Our Contributions}

We begin by studying additive buyer valuation functions. In the non budgeted world of Babaioff et al.~\cite{BNL14} and Lev et al.~\cite{LOBR15} this is an uninteresting case where there is a simple and unique equilibrium in which every item is sold. We show, naturally, that a budgeted buyer can no longer necessarily purchase all of the items for sale. We identify a condition on the valuation function and the budget which is sufficient to ensure a market clearing equilibrium exists. This equilibrium is interesting since each item is priced to provide identical utility to the buyer. We prove that under this condition the market clearing equilibrium is unique. We next examine what happens when the uniqueness condition does not hold in the additive case. We provide a method to identify a subset of the sellers whose existence in some sense generalizes the condition for 
uniqueness and thereby allows for a partial characterization of all PNE. We finally prove our condition for market clearance is necessary.

We then extend our equilibrium analysis to the case of monotone submodular valuation functions. We extend the notion of our conditions in the additive case to submodular valuations proving that a similar market clearing equilibrium exists. We also prove that the condition is again necessary to ensure market clearance. We next provide an example which shows a rather surprising consequence of the budget. We show that at certain equilibrium a buyer who has no budget can pay less and receive more items than their budgeted counterparts. We also show that our notion of generalizing the condition on the valuation and budget to reach an equal-utility equilibrium does not work for XOS valuations. 

Finally, we explore the budget effects on social welfare. We show that while it is known \cite{BNL14}  that without a budget the price of anarchy is 1, with a budget the price of anarchy can be arbitrarily high.


\subsection{Previous work}

As already mentioned, the study of market pricing and equilibria is 
one of the most classical areas of Economics. In particular, the 
study of Walrasian equilibria considers the question of which markets 
have market clearing equilibria; that is, an assignment of prices
to items such that when all agents take their preferred allocation in this pricing (i.e. 
an allocation in their demand set), all items are sold. We are interested especially in settings where there are distinct sellers and buyers, which are termed Fisher markets~\cite{BS05,GJTV07}. A lot of work has gone into examining this when all items are divisible (e.g., commodities, such as oil, grains, etc.), but we focus on the case of indivisible items, which is how most items are sold in say online markets.

One of the foundational ``modern'' papers in this regard is that of Gul and 
Stacchetti \cite{GulS99} who show that for indivisible items  the class of Gross Substitutes, a strict subfamily of submodular functions, is the largest class of functions containing unit demand buyers that always possess a Walrasian equilibrium. There
have been several papers that follow this work 
\cite{GJTV07,RastogiC07,AzarBJ10,BabaioffLNL14,Avigdor-ElgrabliR14,BranzeiCDFF014}. 
The basic emphasis in these papers 
is the existence and
convergence to such equilibria without explicit representation of a pricing
function. Moreover, while budgets are mentioned in  some of these papers, the
usual assumption is that the budgets are sufficiently large. 

Recently, as the ability of vendors to analyze individual buyers and personalize their prices in online markets has increased~\cite{WM01}, research in analyzing competitive pricing analysis in this scenario has increased. The model of Babaioff et al.~\cite{BNL14} is most similar to ours, with many strategic vendors each holding a single item, and a single buyer. The main difference from their work is that while they consider buyers with unlimited purchasing power, we limit ourselves to budget constrained buyers. They prove that in any game with monotone buyer valuations there always exists a pure Nash equilibrium that clears the market and maximizes social welfare. Furthermore, they show that for buyers with submodular valuations there is a unique equilibrium that maximizes social welfare, and that the price of anarchy (and hence price of stability) is 1.

While we do not apply the budgetary constraint to their setting, Lev et al.~\cite{LOBR15} generalized \cite{BNL14} by allowing sellers to sell multiple items, and letting them choose which items to offer and for what price. They prove that an equilibrium is not guaranteed to exist for submodular valuations. When an equilibrium does exist the price of anarchy and price of stability are roughly $\log n$ where $n$ is the number of items offered.

\section{Preliminaries}

Our setting has a set of vendors $N$ ($|N|=n$) in which vendor $i$ sells a single distinct item $i$. A strategy profile of the vendors is a price vector $\textbf{p}=(p_1,\ldots,p_n) \in \mathbb{R}_{\geq 0}^n$ where item $i$ is priced at $p_i$ ($p_i \geq 0$). We use $p(S)=\sum_{i \in S} p_i$ to represent the cost to the buyer for purchasing a set $S \subseteq N$. We write $p_{-i}$ to represent the set of prices excluding $p_i$, we write $(p_i,p_{-i})$ to represent the full set of prices $\textbf{p}$. Opposite the vendors we have a single buyer, represented by a valuation function $v()$, who faced with price vector $\textbf{p}$ consumes a set of items which maximises his utility.

\subsection{The Buyer}

The buyer will be represented by a valuation function and a budget. The valuation function $v : 2^{N} \rightarrow \mathbb{R}_{\geq 0}$ gives non negative value to each subset of $N$. We assume the function is normalized and monotone; that is, $v(\emptyset) = 0$ and $v(S) \leq v(T)$ for $S \subseteq T \subseteq N$. The budget $B \in \mathbb{R}_{>0}$ is simply the maximum capital the buyer has available to purchase items; the budget cannot be exceeded. 

We shall discuss 3 types of valuation functions:
\begin{description}
\item[Additive] This function is defined by a vector of per item valuations $v=(v_1,...,v_n)$ and for any $S \subseteq N$ the valuation is $v(S) = \sum_{i \in S} v_i$.

\item [Submodular] This function is characterized by the idea of decreasing marginal utilities. That is, for a submodular function we have $v(S) - v(S \setminus \{a\}) \geq v(T) - v(T \setminus \{a\})$ for $S \subseteq T \subseteq N$ and $a \in S$.

\item [XOS] An XOS (Exclusive Or of Singletons) function is defined by a finite set of additive functions $F=\{f_1,...,f_k\}$ such that  $v(S) = max_{f_i \in F} f_i(S)$ for every  $S \subseteq N$.
\end{description}

The set of additive functions is a strict subset of the set of monotone submodular functions which in turn is a strict subset of the set of XOS functions. It should be noted all of these sets of functions are strictly contained in the set of subadditive functions which are defined as: $v(S) + v(T) \geq V(S \cup T)$ for $S,T \subseteq N$. For a more in depth look at these functions see \cite{LehmannLN06}.

We assume the buyer has a budget bounded quasi-linear utility function. That is, for some bundle $S \subseteq N$ and pricing vector $\textbf{p}$ we have:

\[
 u_b(S,\textbf{p}) =
 \begin{cases} 
 \hfill v(S) - p(S) \hfill & \text{ if $p(S) \leq B$} \\
 \hfill -\infty \hfill & \text{ otherwise} \\
 \end{cases}
\]

To ease reading and notation, we will write $u_b(i, \textbf{p})$ in place of $u_b(\{i\},\textbf{p})$ and omit the $\textbf{p}$ if it is unambiguous. Following notation in \cite{BNL14}, the demand correspondence of a buyer with valuation function $v$ facing pricing vector $\textbf{p}$ is the family of sets that maximize the buyer's utility:

\[
D(v,\textbf{p}) = \{S \subseteq N : u_b(S) \geq u_b(T) \ \ \forall T \subseteq N\}
\] 

The buyer will consume a set $X(v,\textbf{p}) \in D(v,\textbf{p})$. We assume the buyer has access to a demand oracle which allows the buyer to find $X(v,\textbf{p})$. For ease of notation we usually write $X(v,\textbf{p})$ as $X_{\textbf{p}}$ Often we will find $|D(v,\textbf{p})|=1$ in the cases we study. 
A common assumption~\cite{BNL14} in cases where $|D(v,\textbf{p})|>1$ is that the buyer will opt for a set $S\subseteq N$ with the largest size. We call such a buyer maximal.

\subsection{Vendor Utility and Equilibrium}

As mentioned before, the utility of the vendors is simply the payment they receive for their items. That is, if the buyer consumes $X(v,\textbf{p})$ when presented with prices $\textbf{p}$ the utility to vendor $i$ is:

\[
 u_i(X(v,\textbf{p})) =
 \begin{cases} 
 \hfill p_i \hfill & \text{ if $i \in X(v,\textbf{p})$} \\
 \hfill 0 \hfill & \text{ otherwise} \\
 \end{cases}
\]

In a slight abuse of notation we will often omit the $X(v,\textbf{p})$ from $ u_i(X(v,\textbf{p}))$ instead we write $u_i(\textbf{p})$ implicitly assuming the $X(\cdot)$ and $v$ are in place. A pricing vector $\textbf{p}$ is defined to be a pure Nash equilibrium (PNE) if there does not exist any agents that can improve their utility by unilaterally modifying their price. That is, $\nexists i \in N$ such that $u_i(p_i',p_{-i}) > u_i(p_i,p_{-i}) = u_i(\textbf{p})$ for some  $p_i' \neq p_i$. We note that the sellers also have access to value oracle and are able to determine $X(v,\textbf{p})$.\\

We say a price vector $\textbf{p}$ is {\it market clearing} if $X(v,\textbf{p})=N$ and $\min_i p_i >0$. That is every item is bought and each vendor receives positive utility.

The game as described is entirely parametrized by the budget $B$ and valuation profile $v()$. The agents will price themselves presenting a $\textbf{p}$ that maximizes their individual utility. We are interested in studying the properties of stable pricing schemes, i.e., prices that are a PNE. 
In particular, when are the PNE unique, how must agents price their items, and what is consumed by the buyer. Note this is a game of full information:  the agents are assumed to know $B$ and $v$ and can see the prices $\textbf{p}$ posted by each of the agents.

%

\section{Additive Valuations}

We begin our study with a very simple class of buyer valuations, additive functions. 
Recall $v(\cdot)$ is additive if there exists a set of item values $(v_1,...,v_n)$ and $v(S) = \sum_{i \in S} v_i$ for $S \subseteq N$. As mentioned earlier, previous work~\cite{BNL14,LOBR15} does not explicitly study the additive case, as there is a very simple and unique solution in their setting: assuming a maximal non budgeted buyer, the seller holding item $k$ should simply price the item at $v_k$.  Hence, the maximal buyer will purchase the universe of items and social welfare is maximized. Since each item is being consumed, none have an incentive to lower their price. If $p_k < v_k$ then the agent holding $k$ should be able to increase $p_k$ to $v_k$ and still be consumed by the maximal buyer. Now consider the following instance of our budgeted game $(B=1, (v_1,v_2)=(2,\frac{1}{2}))$. It is easy to verify that the only PNE are of the form: $\textbf{p}=(1,x)$, for $x \in \mathbb{R}_{\geq 0}$, and $X_{\textbf{p}}=v_{1}$. The budget has enabled the vendor with the more valuable item to demand the entire budget excluding the other vendor from participating in the game. 

\subsection{A Unique PNE}

We now provide a sufficient condition (later, in Theorem~\ref{Market_Clearing_is_L}, to be shown necessary) on the buyer's valuation to ensure a unique market clearing PNE in which each vendor participates in the game. More specifically we show that under this condition there is a unique PNE that clears the market providing each vendor with positive utility. Throughout this section we assume $V(N) > B$, otherwise this setting reduces to the setting of $\cite{BNL14}$.


\begin{definition}
The \emph{relative valuation constraint} for set $S\subseteq N$, valuations $(v_{1},\ldots,v_{n})$ and budget $B$ is that for each $v_{i}\in S$, $ v_i > \frac{\sum_{v_{j}\in S\setminus\{i\}} v_j - B} {|S|-1}$
\end{definition}

When considering the constraint on the full set of items ($N$) this is equivalent to stating it as a constraint on the budget: for every $i\in N$, $B>\sum_{j\neq i}(v_{j}-v_{i})$.

We can see that this constraint does not imply that the valuation is a gross substitutes valuation. Namely, the additive function $v(S) = \sum_{i \in S} v_i$ for 3 items $(v_1, v_2, v_3)$ with the values $(2,2,2)$ correspondingly and a budget of $1$. The function satisfies the relative valuation constraint and is not gross substitutes: the pricing $(0.2,0.4,0.4)$ allows the buyer to buy all items, while the pricing $(0.3,0.4,0.4)$ forces the buyer to give up $v_{2}$ or $v_{3}$ despite their prices staying the same, hence not gross substitute\footnote{A different example, shown in Lehmann et al.~\cite{LehmannLN06}, is a budget-additive function $v(S) = \min\{B,\sum_{i \in S} v_i\}$ for 3 items $(v_1, v_2, v_3) = (1,1,2)$ and a budget $B=2$, which satisfies the relative valuation constraint and is not gross substitutes: the prices $(0,\frac{1}{2},1)$, in which $v_{1}$ and $v_{2}$ are purchased, in comparison to prices $(1,\frac{1}{2},1)$ in which only $v_{3}$ would be bought.}.

\begin{theorem}
\label{all_high_theorem}
Given an additive valuation profile $(v_1,...,v_n)$, where $\sum_i v_i > B$ and the valuations for $N$ follow the relative valuation constraint for all items, there is a unique PNE $\textbf{p}$ where $\min_i (p_i) > 0$, $\sum_i p_i = B$, $\forall (i,j), v_i - p_i = v_j - p_j$ and $X(v,\textbf{p})=N$.
\end{theorem}

Before we prove Theorem $\ref{all_high_theorem}$ we will find the following lemma useful:

\begin{lemma}
\label{move_in_lemma}
Given an additive valuation profile $(v_1,...,v_n)$ with budget $B$ and sets $A$ and $U$ where $A \subset U \subseteq N$. If for $i \in U$, $v_i > \frac{\sum_{j \in U \setminus \{i\}} v_j - B}{|U|-1}$ then $v_i > \frac{\sum_{j \in A} v_j - B}{|A|}$.

\end{lemma}

\begin{proof}

Let $i = \argmin_{j \in U \setminus A}v_j$, let $D = U \setminus (A \cup \{i\})$.


From the condition on $v_i$:

\begin{align*}
v_i &> \frac{\sum_{j \in U \setminus \{i\}} v_j - B}{|U|-1} \\
&= \frac{\sum_{j \in D} v_j+ \sum_{j \in A} v_j - B} {|U|-1} \\
\implies (|U|-1)v_i - \sum_{j \in D} v_j &> \sum_{j \in A} v_j - B\\
((|U|-1)-|D|)v_i &\geq (|U|-1)v_i - \sum_{j \in D} v_j > \sum_{j \in A} v_j - B\\
\end{align*}

Where the last line is because $v_i$ was chosen to be the smallest value amongst those in $U \setminus A$. So now we have that:

\begin{align*}
((|U|-1)-|D|)v_i &> \sum_{j \in A} v_j - B\\
\implies v_i &> \frac{\sum_{j \in A} v_j - B}{(|U|-1)-|D|}\\
&= \frac{\sum_{j \in A} v_j - B}{|A|}
\end{align*}
\end{proof}

Now we are ready to prove Theorem $\ref{all_high_theorem}$:

\begin{proof}
We will divide the proof into two parts. First we will define a PNE $\textbf{p}$ which meets our criteria. We then will show that this $\textbf{p}$ is the only possible PNE when $N$ satisfies the  \emph{relative valuation constraint}.

$\textbf{Existence of PNE \textbf{p}:}$ Let $p_i = \frac{B + (n-1)v_i - \sum_{j \neq i} v_j}{n}$. First we note since $v_i > \frac{\sum_{j \neq i} v_j - B} {n-1}$ we get that $p_i > 0$. It is also easy to see that $\sum_i p_i = B$.

We now argue $v_i - p_i = v_j - p_j$ for any pair of vendors $i$ and $j$. Fix some vendor $k$:

\begin{align*}
v_k - p_k &= v_k - \frac{B + (n-1)v_k - \sum_{j \neq k} v_j}{n}\\
&= \frac{nv_k - B - (n-1)v_k + \sum_{j \neq k}v_j}{n} \\
&= \frac{\sum_i v_i - B}{n} \\
&= \frac{v(N) - p(N)}{n}
\end{align*}

That is, each seller $i$ is providing exactly $\frac{1}{n}$ of $v(N) - p(N)$, so each item is priced so they provide identical utility to the buyer.

Next we show $v_i > p_i$ for each item $i$. Suppose for contradiction there is some $i$ where $p_i \geq  v_i$. Since $v_i - p_i = v_j - p_j$ for any item $j$ it must be that $p_j \geq  v_j$. Thus we get $p(N) \geq  v(N) > B$, which contradicts $p(N)=B$ by construction. 

Now we can prove $\textbf{p}$ is a PNE:

Because $\sum_i p_i = B$ and $0 < p_i < v_i$ each seller has their item chosen.
 Thus none have an incentive to lower their price. Say seller $i$ increases their price by of $\epsilon > 0$. Since $(p_i + \epsilon) + \sum_{j \neq i} p_j > B$ the buyer will consume at most $n-1$ items. Since $v_i - p_i = v_j - p_j$ for any pair of items $i$ and $j$ we get $v_i - (p_i + \epsilon) < v_j - p_j$, so item $i$ is not consumed leaving its seller with 0 utility. So $i$ should not increase its price. Thus $\textbf{p}$ is a PNE.

$\textbf{Uniqueness of PNE \textbf{p}:}$ Given that the relative valuation constraint holds for $N$, we now argue that 
$\textbf{p}$ is the only PNE. Assume we have another PNE $\textbf{p}'=(p'_{1},\ldots,p'_{n})$. First we show $\sum_i p'_i = B$:

Suppose $\sum_i p'_i < B$, then clearly $X_{\textbf{p}'} = N$ since $\textbf{p}'$ is a PNE. 
Since $\sum_i v_i > B$ it must be that $p'_j < v_j$ for some item $j$. Let $0 < \epsilon \leq \min\{B - (\sum_i p'_i), \frac{v_j - p'_j}{2}\}$. It is easy to see $(p'_j + \epsilon, p'_{-j})$ is a better solution for seller $j$. By the first term in the $\min$ the new pricing vector is not over budget by the second term $v_j < p'_j$ so the buyer still receives positive utility from item $j$. Thus item $j$ is still consumed and we see that seller $j$ should increase their price by $\epsilon$.

Suppose $\sum_i p'_i > B$, then under this pricing scheme there is an agent $i$, $i \not\in X(v,\textbf{p}')$. We now argue $i$ is able to offer a positive price where this item will be picked. For the remainder of this section when referring to $X_{\textbf{p}'}$ we exclude items which are priced at zero. We can assume that $\sum_{j \in X_{\textbf{p}'}} p'_j = B$, if this were not the case seller $i$ could offer a positive price and still be consumed. 

Let $s = \argmin_{j \in X_{\textbf{p}'}} u_b(j)$, that is $s$ is the seller providing the least amount of utility amongst those who are picked. If $v_i > u_b(s)$ seller $i$ can offer a price of $\min\{\frac{v_i - u_b(s)}{2}, \frac{B - \sum_{j \in X_{\textbf{p}'} \setminus s} p'_j} {2}\}$. The first term in the $\min$ ensures that the buyer gets more utility from $i$ than $s$, the second term ensures the buyer can afford the current chosen set when $s$ will be replaced by $i$. Thus vendor $i$ can replace vendor $s$ in the chosen set.

We now prove $v_i > u_b(s)$. We know that $u_b(s)$ is at most the average utility provided by a seller in $X_{\textbf{p}'}$. That is $u_b(s) \leq \frac{\sum_{j \in X_{\textbf{p}'}} v_j - B}{|X_{\textbf{p}'}|}$. So we must show show $v_i > \frac{\sum_{j \in X_{\textbf{p}'}} v_j - B}{|X_{\textbf{p}'}|}$. This is simply a direct result of Lemma~\ref{move_in_lemma}, where we take the sets $X(v,\textbf{p}')$ and $N$ to be $A$ and $U$ in the lemma respectively.

Thus at any PNE $\textbf{p}'$ each seller offers a positive price is bought and the budget is consumed.

We finally argue $u_b(i) = u_b(j)$ for all sellers $i$ and $j$ at any PNE $\textbf{p}'$. Let $i = \argmax_k u_b(k)$ and $j = \argmin_k u_b(k)$. Assume, for contradiction, $u_b(i) > u_b(j)$. Let $0 < \epsilon \leq \min\{\frac{B - \sum_{k \neq j} v_k}{2}, \frac{u_b(i) - u_b(j)}{2}\}$ and consider the point $(p'_i + \epsilon, p'_{-i})$. Since $\textbf{p}'$ is a PNE, $\sum_i p'_i = B$ so under $(p'_i + \epsilon, p'_{-i})$ at least one item will not be consumed. By the first term in the $\min$, $p'_i + \epsilon + \sum_{k \neq i, j} p'_k < B$, so the buyer can afford all of the items except for $j$. By the second term in the $\min$, $v_i - (p'_i + \epsilon) = u_b(i) - \epsilon > u_b(j)$ so the buyer gets more utility from seller $i$ than from seller $j$. Since seller $j$ is still providing the least utility, it will be the one left out of the chosen set. Thus we see that seller $i$ should increase the price by $\epsilon$ and that $\textbf{p}'$ is not a PNE. Thus it must be that $u_b(i) = u_b(j)$ for all sellers $i$ and $j$.

Thus when the \emph{relative valuation constraint} holds for $N$ we get that all PNE must take the form of $\textbf{p}$ as defined above.

\end{proof}

\subsection{Non Market Clearing PNE and a Complete Characterization of 
Market Clearing PNE}

In this section we will explore what happens when the \emph{relative valuation constraint} does not hold for $N$. Without loss of generality, let the elements of $N$ be in decreasing order of valuations; i.e. if $i < j$ then $v_i \geq v_j$.


\begin{definition}
A \emph{PNE base set} is a set $L\subseteq N$, which is constructed by ordering the elements of $N$ in non-increasing order of valuations. Starting with the first (most valuable) element of $N$ iteratively add the next element to $L$ until reaching an element $i$ s.t. $v_i \leq \frac{\sum_{j \in L} v_j - B}{|L|}$.

Note we are not counting $v_i$ in the above sum nor the size of $L$.

\end{definition}

In a sense $L$ is a maximal set of the largest items. 
We note that a non-empty $L$ must exist since $L=\{1\}$ certainly satisfies the definition. Furthermore, note that for several equal valued items, if one of them is in $L$, then all of them are in $L$.

We now prove some additional useful facts about $L$:

\begin{lemma}
\label{large_value_lemma}
$\sum_{i \in L} v_i > B$
\end{lemma}

\begin{proof}

If $L=N$ then since we assumed $\sum_i v_i > B$ we are done.

Let us suppose $L \neq N$ then there must be an element $s$ s.t. $v_s \leq \frac{\sum_{j \in L} v_j - B}{|L|} $. Rearranging we get that $|L|v_s \leq \sum_{j \in L} v_j - B$, and since $v_s > 0$ we get that $\sum_{j \in L} v_j - B > 0$. Thus we see that $L$ contains strictly more value than the budget.

\end{proof}

\begin{lemma}
\label{relative_large_lemma}
$\forall j \in L, v_j > \frac{\sum_{k \in L \setminus j} v_k - B}{|L|-1}$
\end{lemma}

\begin{proof}

Let $i$ be the last element added to $L$, that is $v_i$ is the smallest value amongst items in $L$. When we added $i$ to $L$ it was because $v_i > \frac{\sum_{j \in L'} v_j - B}{|L'|}$ where $L'$ is $L \setminus \{i\}$. Since $|L| -1= |L'|$ we get $v_i > \frac{\sum_{j \in L \setminus i} v_j - B}{|L|-1}$. 
Now trade $v_i$ with any element in the sum; the value on the left of the inequality can't decrease since $v_i$ was minimal and the value on the right can't increase since we are swapping in $v_i$ for some larger element $v_j$. Thus the inequality still holds. 

\end{proof}

It is easy to see that $L$ as defined above is the unique maximal set such that $\forall i \in L, v_i > \frac{\sum_{j \in L \setminus i} v_j - B}{|L|-1}$. 
If this was not the case and some set $M$ existed s.t. $\forall i \in M, v_i > \frac{\sum_{j \in M \setminus i} v_j - B}{|M|-1}$ and $L \subsetneq M$. Lemma~\ref{move_in_lemma}, where $L$ and $M$ take the place of $A$ and $U$ in the lemma, implies that for each $i \in M$ $v_i > \frac{\sum_{j \in L} v_j - B}{|L|}$. Thus when constructing $L$ and considering element $(|L|+1) \in M$ we should have added it to $L$. Thus $L$ is the largest set for which the \emph{relative valuation constraint} holds.

\begin{lemma}
There is always a PNE $\textbf{p}$ where $X_{\textbf{p}}=L$ and $v_i - p_i = v_j - p_j$ for all pairs of elements in $L$.
\end{lemma}

\begin{proof}

Let $\textbf{p}=(p_1,...,p_{|L|},0...,0)$ where \\ $p_i = \frac{B + (|L|-1)v_i - \sum_{j \in L \setminus i} v_j}{|L|}$. 

By Lemma~\ref{relative_large_lemma} we know $p_i > 0$. By the construction of each $p_i$, $\sum_i p_i = B$ and $v_i - p_i = v_j - p_j = \frac{\sum_{j \in L} v_j - B}{|L|}$ for all pairs $i,j \in L$. As in Theorem~\ref{all_high_theorem} no one agent in $L$ has an incentive to lower their price since they are being sold. Furthermore, if one of the agents in $L$ were to increase their price, they would be removed from $X_{\textbf{p}}$ as they now provide the least utility in $L$, and prices are now over budget. Finally combining Lemma~\ref{large_value_lemma} and a similar argument found in Theorem~\ref{all_high_theorem} we get that $p_i \leq v_i$ for each $i \in L$.

Now we argue that none of the agents in $N \setminus L$ have an incentive to increase their price. Let $v_s$ be the largest value amongst the agents in $N \setminus L$. By the construction of $L$, $v_s \leq \frac{\sum_{j \in L} v_j - B}{|L|}$. If agent $s$ were to offer a price $p_s > 0$ it would not be chosen since $v_s - p_s < \frac{\sum_{j \in L} v_j - B}{|L|}$. Since we are now over budget and the utility provided by $s$ is lower than the utility provided by any of the agents in $L$, $s$ is not chosen. Since $v_s$ was the largest value amongst the agents not in $X_{\textbf{p}}$ none of the agents outside of $X_{\textbf{p}}$ have a sufficient valuation to be chosen.

\end{proof}

Unfortunately we do not get a nice generalisation of Theorem~\ref{all_high_theorem} where at all PNE $v_i - p_i = v_j - p_j$ for each pair $i,j \in L$.

\begin{example}\label{PneNotL}
Let the budget $B=1$, the valuations $\textbf{v} = (2, 1.5, 0.6, 0.6)$ and finally the prices $\textbf{p}=(0.6, 0.4, 0.3, 0.3)$. The utility gained from each agent is simply $\textbf{u}=(1.4, 1.1, 0.3, 0.3)$. It is easy to see that $L = \{1,2\}$ and this is a non market clearing PNE where $X=L$. If the first agent increased their price the new consumed set would be $\{2,3,4\}$. If the second agent increased their price the new consumed set would be $\{1,3\}$ or $\{1,4\}$. And if the third or fourth agent dropped their price to any positive value the consumed set would still be $L=\{1,2\}$.
\end{example}

We can, however, obtain a partial characterization showing that at any PNE $\textbf{p}$ the base set $L$ must be a part of $X_{\textbf{p}}$.

\begin{lemma}
\label{L_in_X_lemma}

Let $\textbf{p}$ be some PNE and $X_{\textbf{p}}$ be the set consumed (not including those selling for free) it must be that $L \subseteq X_{\textbf{p}}$.

\end{lemma}

\begin{proof}

Assume for contradiction that $L \not\subseteq	X_{\textbf{p}}$ at some PNE $\textbf{p}$
We break this into two cases.

Case 1 ($X_{\textbf{p}} \cap (N \setminus L) \neq \emptyset$):

That is, there are elements of $N \setminus L$ in $X_{\textbf{p}}$. Let $s$ be the element of $X_{\textbf{p}} \cap (N \setminus L)$ providing the least utility to the buyer. Let $i$ be some element of $L \setminus X_{\textbf{p}}$. Since $L$ contained the largest valued elements of $N$ it must be that $v_i > v_s$. 
The inequality must be strict for if $v_i = v_s$, $s$ would be a member of $L$. Letting $p_i=p_s$, at this price agent $i$ provides strictly more utility than agent $s$, who is providing the least, so the buyer will replace $s$ with $i$ in the consumed set. So $\textbf{p}$ is not a PNE.

Case 2 ($X_{\textbf{p}} \cap (N \setminus L) = \emptyset$):

That is, $X_{\textbf{p}}$ is a strict subset of $L$. By making an argument similar to the one in Theorem~\ref{all_high_theorem} using Lemma~\ref{move_in_lemma}, with $X_{\textbf{p}}$ and $L$ replacing A and U in the lemma, we get $\textbf{p}$ is not a PNE.

\end{proof}

It is, however, possible that at a PNE \textbf{p},  $X_{\textbf{p}} \neq L$. 
\begin{example}
Let the budget $B=1$, the valuations $\textbf{v} = (2.5, 1.5, 1.4)$ and the prices $\textbf{p}=(0.9, 0.1, 0.9)$. The per agent utility is $\textbf{u}=(1.6, 1.4, 0.5)$. Obviously $L = \{1\}$ and $X_{\textbf{p}} = \{1,2\} \neq L$. If agent 1 raises the price the new consumed set will be $\{2,3\}$, if agent 2 increases the price the new set will be $\{1\}$ and finally if agent 3 drops the price the consumed set will remain $\{1,2\}$. Thus $\textbf{p}$ is a PNE with $L \neq X_{\textbf{p}}$.
\end{example}


We conclude our discussion of additive valuations by proving that 
having a market clearing PNE requires that the relative valuation constraint 
holds for $N$  and hence the PNE is unique and must be of the form described in Theorem~\ref{all_high_theorem}. 

\begin{theorem}
\label{Market_Clearing_is_L}

Consider an additive valuation function $\textbf{v}$ and budget $B$ for
the buyer. 
Let $\textbf{p}$ be some market clearing PNE, such that is $X_{\textbf{p}} = N$ and $\min_i p_i > 0$. It must be that \emph{relative valuation constraint} holds for 
$\textbf{v}$ and $B$; equivalently, the PNE base set $L = N$.  

\end{theorem}

\begin{proof}

Assume, for contradiction, $L \neq N$. Let $s$ be the element providing the least utility to the buyer under the pricing $\textbf{p}$. 

Since $L \subsetneq X_p$ there must be an element $i \in L$ where $v_i - p_i > \frac{\sum_{j \in L}v_j -B}{|L|}$. This is because when the consumed set is only $L$ the average utility provided by an agent in $L$ is $\frac{\sum_{j \in L}v_j -B}{|L|}$. Now we have agents outside of $L$ receiving utility, to accommodate them at least one agent in $L$ must offer a lower price providing more utility. By the construction of $L$ each element of $N \setminus L$ has value at most $\frac{\sum_{j \in L}v_j -B}{|L|}$. Since $s$ is providing the least utility over all agents, and no agent outside of $L$ can provide at most $\frac{\sum_{j \in L}v_j -B}{|L|}$ utility, $v_s-p_s < \frac{\sum_{j \in L}v_j -B}{|L|}$. So we get $v_i - p_i > v_s - p_s$. Thus we see that agent $i$ can increase the price by $\min\{\frac{(v_i - p_i)-(v_s-p_s)}{2}, B - \sum_{j \in N \setminus s}p_j\}$. 
The first term ensures agent $i$ still provides more utility than agent $s$ while the second ensures $N \setminus s$ is budget feasible. Under the new pricing scheme at most $n-1$ items will be bought. Since $s$ is still providing the least utility we get that $N \setminus s$ is the new consumed set. Thus $\textbf{p}$ was not a PNE, a contradiction.
\end{proof}

This theorem, alongside Example~\ref{PneNotL}, means that the Babaioff et al.~\cite{BNL14} result showing that there is always a market clearing equilibrium for {\it all} montone valuation is not true when introducing budgets to the model and even fails for some additive valuations. 

From Theorem~\ref{Market_Clearing_is_L} we conclude that any market clearing PNE \textbf{p} must be of the form of the one described in Theorem~\ref{all_high_theorem}. That is:

\begin{corollary}

If a PNE \textbf{p} is market clearing, then it is unique, and for all pairs $(i,j)$ $v_i-p_i = v_j-p_j$ and $\sum_i p_i = B$.

\end{corollary}

\section{Monotone submodular valuation functions}

We now focus on submodular valuation functions. Recall submodular valuations are characterized by diminishing marginal values, that is for $S \subseteq T$ and $x \not\in T$ we have $v(S \cup \{x\}) - v(S) \geq v(T \cup \{x\}) - v(T)$.

Submodular valuation functions, without budgets, are the main focus of Babaioff et al~\cite{BNL14} and Lev et al.~\cite{LOBR15}. In \cite{BNL14} the authors show that for submodular monotone valuations all PNE are of the form $p_i = v(N) - v(N \setminus \{i\})$. Under this pricing scheme the consumed set will be $S= \{ i : v(N) - v(N \setminus \{i\})> 0\}$. In \cite{LOBR15} the authors show when vendors hold many items there is not necessarily a PNE.

We will show how the introduction of budgets changes the game in similar ways to the addition of budgets in the additive case. We will also see how some non intuitive changes are introduced.

\section{A Complete Characterization of Market Clearance}

In this section we partially generalize Theorem~\ref{all_high_theorem} and Theorem~\ref{Market_Clearing_is_L} from the additive case identifying necessary and sufficient conditions for market clearance in the submodular case.

Define $v_T(S) = v(T) - v(T \setminus S)$ for $S \subseteq T \subseteq N$,  that is the marginal value of set $S$ for completing set $T$. To simplify reading, for a set $S$ with single item $s$, we will use the notation $v_{T}(s)$ to denote $v_{T}(\{s\})$. We also define $u_b^i(S) = v_S(i) - p_i$, that is the seller $i$'s marginal contribution to the buyers utility for completing set $S$.
We will find the following lemmas regarding submodular and monotone functions useful in this section:

\begin{lemma}
\label{SM helper lemma 1}
For a submodular valuation $v()$ and pricing scheme $\textbf{p}$, if for all vendors $i$, $p_i \leq v_D(i)$ then $v(D) - p(D) \geq v(C) - p(C)$ for sets $C \subseteq D$
\end{lemma}

\begin{proof}

Consider some seller $j \in D \setminus C$. Since $v$ is submodular and $p_j \leq v_D(j)$ we get $p_j \leq v_{C \cup \{j\}}(j)$, so $v(C \cup \{j\}) - p(C \cup \{j\}) \geq v(C) - p(C)$. Since $C \cup \{j\} \subseteq D$ we can take another seller $j' \in C \cup \{j\} \setminus D$ and again conclude due to $p_j' \leq v_D(j')$ and the submodularity of $v$ $v(C \cup \{j,j'\}) - p(C \cup \{j,j'\}) \geq v(C) - p(C)$. We can repeat this process for the rest of the elements in $C \setminus D$ concluding $v(D) - p(D) \geq v(C) - p(C)$ for sets $C \subseteq D$.

\end{proof}

\begin{lemma}
\label{monotone helper lemma 1}
For a monotone valuation $v$ and pricing scheme $\textbf{p}$, if for vendors $a$ and $c$ $v_N(a) - p_a = v_N(c) - p_c$ then $v(N \setminus c) - p(N \setminus c) = v(N \setminus a) - p(N \setminus a)$
\end{lemma} \begin{proof}

Assuming $v_N(a) - p_a = v_N(c) - p_c$:

\begin{align*} 
& v_N(a) - p_a = v_N(c) - p_c\\
\implies & v(N) - v(N \setminus a) - p_a = v(N) - v(N \setminus c) - p_c\\
\implies & v(N \setminus c) - p_a = v(N \setminus a) - p_c \\
\implies & v(N \setminus c) - p_a - \sum_{k \neq a,c}p_k= v(N \setminus a) - p_c - \sum_{k \neq a,c}p_k \\
\implies & v(N \setminus c) -\sum_{i \neq c} p_i = v(N \setminus a) - \sum_{i \neq a} p_i
\end{align*}\end{proof}

\begin{theorem}
\label{submodular_all_high_theorem}

Given a submodular valuation function $v$ where $\forall i \in N, v_N(i) > \frac{\sum_{j \neq i} v_N(j) - B}{n-1}$ there is a market clearing PNE $\textbf{p}$ where $\sum_i p_i = B$ and $v_N(i) - p_i = v_N(j) - p_j$ for all pairs of vendors $i$ and $j$.
\end{theorem}

\begin{proof}
Let $p_i = \frac{B + (n - 1)v_N(i) - \sum_{j \neq i} v_N(j)}{n}$, this value is positive because of the condition on $v_N(i)$. As in the proof of Theorem~\ref{all_high_theorem} it is easy to see $v_N(i) - p_i = \frac{\sum_{j \in N} v_N(j) - B}{n}$ so $v_N(i) - p_i$ is the same for each vendor $i$. It is also easy to see $\sum_i p_i = B$. Combining these two facts we see that as in the proof of Theorem~\ref{all_high_theorem} $p_i < v_N(i)$. Using these facts and applying Lemma $\ref{SM helper lemma 1}$ to $N$ we see each agent must be chosen and none have an incentive to lower their price.

Suppose $\textbf{p}$ is not a PNE It must be that some agent has an incentive to increase their price, call this agent $a$. If $a$ increases its price by $\epsilon > 0$, since $\sum_i p_i = B$ at least one of the agents must not be chosen under the new pricing scheme. Call this other agent $c$. 
Let $\textbf{p}'$ be the new pricing scheme where $a$ increases its price.

First we argue that $v(N \setminus c) - \sum_{i \neq c} p_i > v(X_{\textbf{p}'}) - \sum_{i \in X_{\textbf{p}'}} p_i -\epsilon$. Note $v(X_{\textbf{p}'}) - \sum_{i \in X_{\textbf{p}'}} p_i - \epsilon < v(X_{\textbf{p}'}) - \sum_{i \in X_{\textbf{p}'}} p_i$. This is now a simple application of Lemma $\ref{SM helper lemma 1}$ where $C$ and $D$ in the lemma are replaced with $X_{\textbf{p}'}$ and $N \setminus c$ respectively. Thus we get $v(N \setminus c) - \sum_{i \neq c} p_i \geq v(X_{\textbf{p}'}) - \sum_{i \in X_{\textbf{p}'}} p_i > v(X_{\textbf{p}'}) - \sum_{i \in X_{\textbf{p}'}} p_i -\epsilon$.

Next we get directly from Lemma \ref{monotone helper lemma 1} $v(N \setminus c) -\sum_{i \neq c} p_i = v(N \setminus a) - \sum_{i \neq a} p_i$. Thus we get $v(N \setminus a) - \sum_{i \neq a} p_i = v(N \setminus c) - \sum_{i \neq c} p_i > v(X_{\textbf{p}'}) - \sum_{i \in X_{\textbf{p}'}} p_i -\epsilon$. Since $\sum_{i \neq a}p_i < B$ we see that the new consumed set will be all agents except $a$; thus it should not increase the price. So we see $\textbf{p}$ was a indeed an equilibrium.
\end{proof}

Similar to Theorem~\ref{Market_Clearing_is_L} from the additive case we show that in the submodular case the only market clearing PNE are of the equal marginal utility kind. That is the all market clearing PNE must be of the form of the one described in Theorem~\ref{submodular_all_high_theorem}

\begin{theorem}\label{submodularUnique}

For a submodular valuation $v$ given any PNE $\textbf{p}$ where $X(v,\textbf{p}) = N$ and $u_i > 0$ for all vendors $i$, it must be that $v_N(k) - p_k = v_N(j) - p_j$ for any pair of vendors $k$ and $j$.

\end{theorem}

\begin{proof}

Suppose at PNE $\textbf{p}$ $X_{\textbf{p}} = N$. Let \\ $a = argmax_{i \in N} v_N(i) - p_i$, that is, $a$ is providing the most marginal utility, so $v_N(a) - p_a \geq v_N(i) - p_i$ for every other vendor $i$. Assume for contradiction, that in at least one case this inequality is strict. Let vendor $c = \argmin_i v_N(i) - p_i$ so $v_N(a) - p_a > v_N(c) - p_c$.

We now argue $a$ has an incentive to increase its price by $\epsilon < \min\{\frac{u_b^a(N)}{2},\frac{u_b^a(N) - u_b^c(N)}{2}, \min_{j \neq a} B - p(N \setminus j)\}$ (recall $u_b^i(N) = v_N(i) - p_i$). The first term is positive since since  $v_N(c) - p_c \geq 0$ thus $v_N(a) - p_a > 0$ so $v_N(a) > p_a$. The second positive since $v_N(a) - p_a > v_N(c) - p_c$. The third is positive because $\sum_i p_i \leq B$. Let $(p_a + \epsilon , p_{-a}) = \textbf{p}'$.

Since $X_{\textbf{p}} = N$ we know that $p_i \leq v_N(i)$ for each vendor $i$ and because of the first term in the above $\min$, $p'_i \leq v_N(i)$. We can apply lemma \ref{SM helper lemma 1} taking set D in the lemma to be any $K \subset N$ where $|K| = |N| - 1$ getting $v(K) - p'(K) \geq v(J) - p'(J)$ for $J \subseteq K$. So we see under $\textbf{p}'$ the buyer will consume a set of size $|N|-1$, assuming it is under budget.

By the second term in the $\min$, $v_N(a) - p_a - \epsilon > v_N(c) - p_c$. Using this we can modify the proof of Lemma  \ref{monotone helper lemma 1}, replacing the equalities with inequalities, getting $v(N \setminus a) - \sum_{i \neq a} p_i < v(N \setminus c) - \sum_{i \neq c} p_i - \epsilon$. Thus the buyer will prefer to exclude $c$ instead of $a$ from $X_{\textbf{p}'}$, assuming the chosen set with $a$ is under budget. The third term in the $\min$ ensures any set of size $|N| -1$ that $a$ is a part of will remain under budget.

Thus we see $a$ can raise its price by $\epsilon$ and $\textbf{p}$ can't be a PNE.

\end{proof}

We note that we cannot extend the idea of the PNE base set to the submodular case in the same way -- while a maximal set of items that adhere to the condition can be found, as shown in Example~\ref{2L}, the set is not unique, and hence not all equilibria are an extension of the same set.

\begin{example}\label{2L}

We have a set of 4 items $N=\{a,b,c,d\}$. $v(a)=v(b)=1$; $v(c)=v(d)=\frac{1}{2}$. $v(\{a,b\})=\frac{3}{2}$, $v(\{c,d\})=1$, $v(\{a,c\})=\frac{3}{2}$. $v(\{a,b,c\})=\frac{7}{4}$, $v(\{a,c,d\})=1.52$ and $v(\{a,b,c,d\})=1.76$. The values for the rest of the sets are defined by this as $a$ and $c$ can be replaced with $b$ and $d$, respectively (if they were not in the original set). The budget is $0.3$.

Both sets $\{a,b,c\}$ and $\{a,b,d\}$ conform to the relative valuation constraint (their marginal values with respect to the set are equal). However, $a,b,c,d$ does not.

\end{example}

\subsection{When the Sum of Marginal Values is Below the Budget}

When the sum of marginal values of items is below the budget, the equilibrium shown in \cite{BNL14} is both market clearing and lets the buyer keep some of their money. However, that equilibrium is unique only when the budget is not known to the sellers. When it is known, other equilibria arise, and some of them provide less utility to the buyer than the market clearing one, and reduces the social welfare. This shows buyers are worse off announcing their budget, as Example~\ref{budgetIsBad} shows.

\begin{example}\label{budgetIsBad}

We have a set of 4 items $N=\{a,b,c,d\}$. $v(a)=v(b)=1$; $v(c)=v(d)=\frac{1}{4}$. $v(\{a,b\})=\frac{7}{4}$, $v(\{c,d\})=\frac{1}{2}$, $v(\{a,c\})=\frac{9}{8}$. $v(\{a,b,c\})=1.76$, $v(\{a,c,d\})=1.135$ and $v(\{a,b,c,d\})=1.77$. The values for the rest of the sets are defined by this as $a$ and $c$ can be replaced with $b$ and $d$, respectively (if they were not in the original set). The budget is $1$.

The marginal values (in $N$) of $a$ and $b$ are $0.27$, of $c$ and $d$ $0.1$, so pricing the items with these values is both under budget and an equilibrium, as \cite{BNL14} claim, and will result in all items being purchased by the seller. However, another equilibrium is pricing $a$ and $b$ at $\frac{1}{2}$ each, and $c$ and $d$ at $\frac{1}{4}$. Under these prices, the buyer will only buy items $a$ and $b$ (as they exhaust the budget). This is an equilibrium, as neither $a$ nor $b$ can raise their prices, as that precludes buying both of them, and buying the one that didn't raise its price with $c$ and $d$ is more beneficial for the buyer. Items $c$ and $d$ can't lower their prices to anything above $0$, as there is no budget to purchase them, and even with only one of $a$ or $b$, they still do not provide as much utility as purchasing both $a$ and $b$.
\end{example}

Finally, as Example~\ref{xosNoEqui} shows, unlike submodular valuation functions, we cannot always guarantee equal utility equilibria for XOS functions even when the valuations adhere to the relative valuation constraint:

\begin{example}\label{xosNoEqui}
We have a set of 3 items $N=\{a,b,c\}$. The XOS valuation function is defined using 2 additive valuations: The first is $(2,1,1)$ respectively for $(a,b,c)$, and the second is $(3,0,0)$. The budget is $1.5$. The marginal value of item $a$ is $2$, and $b$ and $c$ have marginal value $1$, so they comply with the relative valuation constraint.

An equal marginal utility contribution would mean pricing item $a$ at $\frac{7}{6}$ and items $b$ and $c$ at $\frac{1}{6}$ each. However, in this case item $a$ can increase its price to $\frac{4}{3}$, as without it, by buying only $b$ and $c$ the buyer has a utility of $\frac{5}{3}$, but by buying $a$ and $b$, the buyer gets a utility of $2.5$, being obviously better off, so this is not an equilibrium.
\end{example}

\section{Social Welfare}

While we have shown price equilibria in various settings, we have yet to investigate the value of the sold set to the buyer. As the sellers get the money from the buyer, this is equivalent to analyzing the social welfare -- the amount of utility ``created'' in the equilibrium.

The common measures of the impact of equilibrium on the welfare are price of anarchy (PoA), and, to a lesser extent, price of stability (PoS), which are, respectively, the ratio between the social welfare in the worst (respectively, the best) equilibria and the optimal welfare. However, the optimal welfare is, of course, that all items are given to the buyer for free, and therefore, the social welfare is the value of the whole set, namely $v(N)$. 

In \cite{BNL14} the authors show that for monotone submodular valuations the PoA is always 1. Furthermore they show for any monotone function the PoS is 1, although the PoA can be infinite. In the multi item setting of \cite{LOBR15} the authors show when PNE do exist for monotone submodular valuations both the PoA and PoS are approximately $log(m)$ where $m$ is the number of items.

The usage of PoA is intended to aid us in understanding the effects of an equilibrium compared to the optimal state. Assuming sellers get nothing is not only unrealistic, it also does not let us consider the effect of an equilibrium under the budget constraints (as the budget plays no role). Therefore, we would like to have some measure that takes into account the existence of a budget, as well as the desire of sellers to be paid.

To deal with this problem we use a somewhat different metric: $\frac{PoA}{PoS}$, i.e., the ratio between the social welfare in the worst equilibrium and the social welfare in the best equilibrium. In order to avoid items being ``given away'' and artificially adding to the utility, we do not count the social welfare garnered from items priced at $0$. Our measure shows the differences between different equilibria under the same budgetary constraint, and is, therefore, an indicator of how bad equilibria can become. Surprisingly. and in contrast to the results in \cite{BNL14}, which showed that without a budget the $PoA$ is 1, we discover the $PoA$ can be arbitrarily bad for submodular functions. Moreover, with budgets, even in the additive case, the PoS can be greater than 1, as is shown by Theorem~\ref{Market_Clearing_is_L}, and examples such as Example~\ref{PneNotL}. Therefore, we can have $PoA>\frac{PoA}{PoS}\geq PoS>1$.


\begin{theorem}
The $\frac{PoA}{PoS}$ (and hence the PoA) can be unbounded and approaching 
$n-2$ even for additive valuation functions. A corresponding upper bound is $n$, even for submodular functions.
\end{theorem}
\begin{proof}
We first show that the ratio is unbounded for the additive case and approaching $n-2$, and then show  the upper bound of $n$ for submodular functions.

Let $v$ be an additive function. Let the value of item $1$ be $2$, the value of item $n-1$ and $n$ be $0.55$, and all the rest of the items (from $2$ to $n-2$) have a value of $1$. The budget is $1$. It is easy to see the set $L$ only contains the first item, and therefore there is an equilibrium in which all the budget goes to this item, and all others items priced at $0$. Hence the utility of this equilibrium is $2$ (note that as any equilibrium must include $L$, this is the worst possible equilibrium).

Now consider the following pricing: item $1$ is priced at $\frac{1}{2}$, items valued at $1$ have a price of $\frac{1}{2(n-3)}$, and the final two items are priced at $0.25$ each. This is a Nash equilibrium in which all items but the last two are sold. None of the items valued at $1$ can increase their price (as then they will be the item providing the least utility and be thrown out), and neither can the items valued at $0.55$ lower their price to be purchased, as even then they provide too little utility. If item 1 tries to raise its price (which necessitates throwing out one of the items valued at $1$), it is better for the buyer not to purchase item 1, and instead add both items $n-1$ and $n$ to the purchases set. The value of this equilibrium is $2+(n-3)$.  Thanks to Theorem~\ref{Market_Clearing_is_L}, we know there is no market clearing equilibrium, and as is shown in the next paragraph, there can be no equilibrium with $n-1$ items.  Hence $\frac{PoA}{PoS}= \frac{2+(n-3)}{2}$, which is unbounded.

Notice that replacing the valuation of the top item to any $k\in\mathbb{N}$ and the next $n-3$ items to $k-1$, does not change the existence of both equilibria, and $\frac{PoA}{PoS}\geq \frac{k(n-2)-0.1}{k+1}\xrightarrow[k\rightarrow\infty]{} n-2$

Looking at an upper bound, note that $\frac{PoA}{PoS}$ is actually $\frac{\text{value of best NE}}{\text{value of worst NE}}$. The best equilibrium is, at most, the value of the whole set, $v(N)$, and thanks to submodularity, $v(N)\leq n \cdot \max_{i\in N}v(\{i\})$. Furthermore, in any equilibrium, the buyer gets at least a value$\geq \max_{i\in N}v(\{i\})$. If not, suppose $i'$ is the value maximizing $v(\{i\})$, and there is an equilibrium where the buyer purchases a set $S\subseteq (N\setminus \{i'\})$ at price $b\leq B$, and $v(S)<v(\{i'\})$. Vendor $i'$ can improve its situation (a utility of $0$), by changing its price to $b$, and being bought instead of set $S$, proving that buying set $S$ for a price of $b$ is not an equilibrium. 
Hence the utility of the worst equilibrium is at least $\max_{i\in N}v(\{i\})$. Therefore,  $\frac{PoA}{PoS}\leq n$.

\end{proof}

\section{Conclusion and Future Work}

The addition of budgets to the work of ~\cite{BNL14} introduces interesting and strategic behaviour amongst agents. We identify the  \emph{relative valuation constraint} and prove it is both sufficient and necessary to ensure market clearance with submodular valuations, and in the additive case, we show the market clearing equilibrium is unique. This market clearing equilibrium is conceptually interesting, as in it each vendor provides equal marginal utility to the buyer for completing the universe of items. We also show how the constraint does not extend to XOS valuations where an equal marginal utility equilibrium is no longer guaranteed to exist. Furthermore we generalize the \emph{relative valuation constraint} to subsets of the vendors for the additive case, and using this generalization, we identify a maximal base set of items (values) that can form an equilibrium similar to the market clearing one excluding other sellers. Finally we provide examples illustrating how the budget can influence the game in unexpected ways. Perhaps most striking of which is example \ref{budgetIsBad} which can be interpreted as showing that a buyer announcing their budget may induce new and inferior equilibrium where they pay more and receive less compared to not announcing the budget at all.
 

Generalizing our model to many items per seller, as in \cite{LOBR15}, would be a likely logical extension to our work. When sellers hold many items the problem domain is significantly more complex. It is unknown even with simple additive valuations what Nash equilibrium this setting permits. It would also be interesting to see if with a budgeted and submodular buyer we also encounter the issue of no equilibria. Also, while \cite{BNL14,LOBR15} are able to avoid the issue of item cost, due to the lack of budget, this issue becomes more interesting in our setting, and is yet to be explored.

Another direction that requires more attention is the multi buyer setting. The simplicity of our model where there is one idealized buyer lends itself to simple analysis. It is known \cite{GulS99} that with multiple buyers there is no Walrasian equilibrium unless the valuation class is essentially  the 
class of Gross Substitutes \cite{GulS99,BenZviLN13}. However, there may still be interesting results for a bounded, small number of buyers. Another open problem is whether the budgeted model may preclude any Nash equilibria (not just market clearing ones) in certain settings.

We suspect our results for the \emph{relative valuation constraint} should generalize to XOS valuations. While the PNE would not be an equal marginal utility one, as our example shows is impossible, we do believe it is a necessary and sufficient condition for market clearance. Beyond this it would be interesting to study more general valuation functions, even ones that exhibit complements.

All of these directions are reasonable intermediate steps towards the most general pricing model with sellers holding many items and many buyers with complex valuations and differing budgets.

\section{acknowledgements}

The authors thank Brendan Lucier for his insightful comments and clarifications. This work is supported by NSERC grant 482671.

%
%
\bibliographystyle{abbrv}
\bibliography{general} 

\begin{thebibliography}{10}

\bibitem{Avigdor-ElgrabliR14}
N.~Avigdor{-}Elgrabli and Y.~Rabani.
\newblock Convergence of t{\^{a}}tonnement in fisher markets.
\newblock {\em CoRR}, 2014.

\bibitem{AzarBJ10}
Y.~Azar, N.~Buchbinder, and K.~Jain.
\newblock How to allocate goods in an online market?
\newblock In {\em Algorithms - {ESA} 2010, 18th Annual European Symposium,
  Liverpool, UK, September 6-8, 2010. Proceedings, Part {II}}, pages 51--62,
  2010.

\bibitem{BabaioffLNL14}
M.~Babaioff, B.~Lucier, N.~Nisan, and R.~P. Leme.
\newblock On the efficiency of the walrasian mechanism.
\newblock In {\em {ACM} Conference on Economics and Computation, {EC} '14,
  Stanford , CA, USA, June 8-12, 2014}, pages 783--800, 2014.

\bibitem{BNL14}
M.~Babaioff, N.~Nisan, and R.~P. Leme.
\newblock Price competition in online combinatorial markets.
\newblock In {\em Proceedings of the 23rd international conference on World
  Wide Web (WWW)}, pages 711--722, Seoul, Korea, April 2014.

\bibitem{BenZviLN13}
O.~Ben{-}Zwi, R.~Lavi, and I.~Newman.
\newblock Ascending auctions and walrasian equilibrium.
\newblock {\em CoRR}, 2013.

\bibitem{BS05}
W.~C. Brainard and H.~E. Scarf.
\newblock How to compute equilibrium prices in 1891.
\newblock {\em American Journal of Economics and Sociology}, 64(1):57--83,
  January 2005.

\bibitem{BranzeiCDFF014}
S.~Br{\^{a}}nzei, Y.~Chen, X.~Deng, A.~Filos{-}Ratsikas, S.~K.~S. Frederiksen,
  and J.~Zhang.
\newblock The fisher market game: Equilibrium and welfare.
\newblock In {\em Proceedings of the Twenty-Eighth {AAAI} Conference on
  Artificial Intelligence, July 27 -31, 2014, Qu{\'{e}}bec City, Qu{\'{e}}bec,
  Canada.}, pages 587--593, 2014.

\bibitem{RastogiC07}
R.~Cole and A.~Rastogi.
\newblock Indivisible markets with good approximate equilibriumprices.
\newblock {\em Electronic Colloquium on Computational Complexity {(ECCC)}},
  14(017), 2007.

\bibitem{GJTV07}
D.~Garg, K.~Jain, K.~Talwar, and V.~V. Vazirani.
\newblock A primal-dual algorithm for computing fisher equilibrium in the
  absence of gross substitutability property.
\newblock {\em Theoretical Computer Science}, 378(2):143--152, June 2007.

\bibitem{GulS99}
F.~Gul and E.~Stacchetti.
\newblock Walrasian equilibrium with gross substitutes.
\newblock {\em Journal of Economic Theory}, 87(1):95--124, July 1999.

\bibitem{LehmannLN06}
B.~Lehmann, D.~J. Lehmann, and N.~Nisan.
\newblock Combinatorial auctions with decreasing marginal utilities.
\newblock {\em Games and Economic Behavior}, 55(2):270--296, 2006.

\bibitem{LOBR15}
O.~Lev, J.~Oren, C.~Boutilier, and J.~S. Rosenschein.
\newblock The pricing war continues: On competitive multi-item pricing.
\newblock In {\em Proceedings of the 29th AAAI Conference on Artificial
  Intelligence (AAAI)}, pages 972--978, Austin, Texas, January 2015.

\bibitem{WalrasJ03}
L.~{Walras (Translation by William Jaffe)}.
\newblock {\em Elements of Pure Economics, Or the Theory of Social Wealth}.
\newblock Taylor and Francis, 2003.

\bibitem{WM01}
R.~M. Weiss and A.~K. Mehorta.
\newblock Online dynamic pricing: Efficiency, equity and the future of
  e-commerce.
\newblock {\em Virgina Journal of Law and Technology}, 6(2), 2001.

\end{thebibliography}
%

\end{document}